\documentclass[twocolumn,letter]{IEEEtran}
\usepackage[usenames,dvipsnames]{color}
\usepackage{subfigure}
\usepackage{graphicx}
\usepackage{amsmath}
\usepackage{color}
\usepackage{multicol}
\usepackage{amssymb}
\usepackage{graphicx}

\usepackage{amsfonts}%
\usepackage{verbatim}%
\usepackage{enumerate}%
\usepackage{psfrag}
\usepackage{epsfig}
\usepackage{multicol}
\usepackage{setspace}
\usepackage{dsfont}

\usepackage{balance}

\usepackage{algorithm}
\usepackage{algorithmic}
\usepackage{cite}
\usepackage{float}
\usepackage{cite}
 \newtheorem{proposition}{Proposition}
\begin{document}
\title{Band Allocation for Cognitive Radios with Buffered Primary and Secondary Users}
\author{ Ahmed El Shafie$^\dagger$, Ahmed Sultan$^\star$, Tamer Khattab$^*$\\
\small \begin{tabular}{c}
$^\dagger$Wireless Intelligent Networks Center (WINC), Nile University, Giza, Egypt. \\
$^\star$Department of Electrical Engineering, Alexandria University, Alexandria, Egypt.\\
$^*$Electrical Engineering, Qatar University, Doha, Qatar. \\
\end{tabular}
}
%\author{ Ahmed El Shafie$^\dagger$, Ahmed Sultan$^*$\\
%\small \begin{tabular}{c}
%$^\dagger$Wireless Intelligent Networks Center (WINC), Nile University, Giza, Egypt. \\
%$^*$Department of Electrical Engineering, Alexandria University, Alexandria, Egypt. \\
%\end{tabular}
%}
\date{}
\maketitle
\thispagestyle{empty}
\pagestyle{empty}
\begin{abstract}
%In this paper, we study band allocation of $\mathcal{M}_s$ buffered secondary users (SUs) to $\mathcal{M}_p$ orthogonal primary licensed bands, where each primary band is assigned to one primary user (PU) operating in a time-slotted fashion and transmitting starting at the beginning of the time slot if its queue is nonempty. The SUs, depending on their queues and spectrum sensing results, transmit after $\tau$ seconds relative to the beginning of the time slot, where $\tau $ is the sensing duration. Each SU is assigned to one of the available primary bands with a certain probability designed to satisfy some specified quality of service (QoS) requirements for the SUs. In the proposed system, only one SU is assigned to a particular band. The optimization problem used to obtain the stability region's envelope (closure) is shown to be a linear program. We compare the stability region of the proposed system with that of a system where each user chooses a band randomly with some assignment probability. We also compare with a fixed (deterministic) assignment system, where only one SU is assigned to one of the primary bands all the time. We prove the advantage of the proposed system over the other systems. We provide some numerical results of the considered systems in this paper.
In this paper, we study band allocation of $\mathcal{M}_s$ buffered secondary users (SUs) to $\mathcal{M}_p$ orthogonal primary licensed bands, where each primary band is assigned to one primary user (PU). Each SU is assigned to one of the available primary bands with a certain probability designed to satisfy some specified quality of service (QoS) requirements for the SUs. In the proposed system, only one SU is assigned to a particular band. The optimization problem used to obtain the stability region's envelope (closure) is shown to be a linear program. We compare the stability region of the proposed system with that of a system where each SU chooses a band randomly with some assignment probability. We also compare with a fixed (deterministic) assignment system, where only one SU is assigned to one of the primary bands all the time. We prove the advantage of the proposed system over the other systems.

% It is shown that the boundary of the proposed systems is a convex set. The delay performance of the proposed systems is characterized.
\end{abstract}
\section{Introduction}
There is a recent dramatic increase in the demand for radio spectrum, stimulated by the enormous influx of new wireless devices and applications. The cognitive radio communications paradigm allows a more effective and efficient use of the electromagnetic spectrum. Cognitive or secondary users (SUs) utilize the spectrum when it is unused by the primary or licensed system. The question arises as to how the SUs access a primary channel while satisfying some quality of service (QoS) specifications. The design of an efficient medium access control (MAC) protocol to assign the SUs to the available primary bands is very crucial.

The problem of band allocation in a cognitive radio setting has been studied in many works \cite{liu2008distributed,liu2008randomized,digham2008joint,lu2009optimal,gai2010learning,lai2011cognitive,shiang2008queuing,queues}. In order to avoid convergence to the same channels, \cite{liu2008distributed} proposes a simple randomized sensing policy where the channel selection probability by each SU is determined by its belief, which is the conditional probability, given all past decisions and observations, that the channels are in a particular state of occupancy by the primary users (PUs). In \cite{liu2008randomized}, the probability to sense each channel is assigned to every SU, and the sensing policy is formulated as an optimization problem over all combinations of the assignment probabilities to maximize the total throughput of the network.
The work in \cite{digham2008joint} investigates the case where a set of channels is distributed among multiple secondary nodes
that opportunistically access the available spectrum. The solution of the band allocation problem is obtained via maximizing the total sum capacity of the cognitive radio network.  By introducing an interference
temperature constraint for guaranteeing PUs' QoS, the authors of \cite{lu2009optimal} proposed an optimal subcarrier and
power allocation algorithm to maximize the overall utility
for SUs. In \cite{lai2011cognitive}, a cognitive medium access protocol is proposed for uncertain environments where the PU traffic statistics are unknown a priori and have to be learned and tracked. In the case of multiple SUs, the channel selection is formulated as an optimization problem for cooperative SUs and a non-cooperative game for selfish SUs, respectively.

The presence of data queues in the system has not been considered in the aforementioned works. Resource allocation involving buffer dynamics in a cognitive setting has been considered in a few works such as \cite{shiang2008queuing} and \cite{queues}. In \cite{shiang2008queuing}, a dynamic channel-selection for autonomous wireless users is proposed, where each user has set of actions and strategies. Based on the priority queueing analysis (i.e., priority classes among SUs), each wireless
user can evaluate its utility impact based on the behaviors of
the users deploying the same frequency channel including the
PUs. The work in \cite{queues} investigates the resource allocation problem for the downlink of an OFDMA-based cognitive radio network. Prior to the beginning of each frame each user transmits to the base station its sensing
information vector as well as its latest channel gain
vector which was obtained based on pilot symbols. Based on the received information from the users and the current backlog for each
user, the base station performs resource allocation for the frame. The resource allocation map is then sent to the users and is
valid for the remainder of the frame, which is composed of multiple time slots.

%Unlike \cite{liu2008randomized}, we consider buffered users. Also, unlike \cite{liu2008randomized,shiang2008queuing}, we include channels outage effects and we consider time slotted channels.
In this work, we consider buffered terminals, time slotted channels and include the impact of channel outage on the system's performance.
We also do not assume the availability of channel side information (CSI) at the transmitting terminals. Specifically, we consider a time-slotted primary channels over which each PU transmits starting at the beginning of the time slot whenever it has packets to communicate. Each PU is assigned singly to one of the bands. In the proposed system, denoted by $\mathcal{S}$, each band has at most one SU. The SUs are probabilistically assigned to the $\mathcal{M}_p$ bands at the beginning of each time slot. When an SU is assigned to a band, it has to sense that band for $\tau$ seconds relative to the beginning of the time slot to detect the activity of the PU which owns that band. Varying the assignment probabilities, we can obtain the maximum stable throughput region for the secondary network.

We make the following contributions in this paper.
\begin{itemize}
\item We study resource allocation in a cognitive radio network with buffered users and propose a channel allocation method for the SUs which results in probabilistic assignment with each primary channel allocated to one SU at maximum.
\item We provide the exact maximum stable throughout region of the proposed system, which is obtained via solving a linear program.
    \item We provide proofs for the advantage of the proposed system, in terms of service rates, over systems with fixed channel assignments and systems where each SU can randomly choose (access) any band and, hence, collisions among the SUs may occur over a primary channel.
\end{itemize}

The rest of the paper is organized as follows. In the following section,
we describe the system model considered in this paper. The stability region of the proposed system is obtained in Section \ref{system_S}. System where each SU can randomly select any band and the fixed assignment system are discussed in Section \ref{comparisonsection}. We present some numerical results for the optimization problems presented in this paper in Section \ref{numerical} and conclude the paper in Section \ref{conc}.
%%%%%%%%%%
%%%%%%%%%%%
%%%%%%%% System Model

\section{system model}
\label{system_model}
We propose a cognitive radio system, denoted by $\mathcal{S}$, in which the SUs are assigned to $\mathcal{M}_p$ licensed orthogonal
frequency bands over which the PUs operate in a time-slotted fashion. The primary band $B_j$ has bandwidth $W_{j}$, where in general $W_j\ne W_i$ for all $j\ne i$ and $j,i\in \{1,2,\dots,\mathcal{M}_p\}$. The secondary network consists of a finite number $\mathcal{M}_s$ of terminals numbered
$1,2,\dots,\mathcal{M}_s$. Each terminal, whether primary or secondary, has an infinite queue for storing fixed-length packets. The $j${\it th} PU, $p_j$, has a queue denoted by $Q_{{p}_{j}}$, whereas the $k${\it th} SU, $s_k$, has a queue denoted by $Q_{s_k}$. We adopt a discrete-time late arrival model, which means that a newly arrived packet during a particular time slot cannot be transmitted during the slot itself even if the queue is empty. Arrival processes at all queues are Bernoulli random variables that are independent across terminals and independent from slot to slot \cite{sadek}. The mean arrival rate at $Q_{p_j}$ is $\lambda_{ p_j}$ and at $Q_{s_k}$ is $\lambda_{s_k}$. If a terminal transmits during a time slot, it sends exactly one packet to its receiver.

A PU, $p_j$, assigned to band $B_j$, transmits the packet at the head of its queue starting at the beginning of the time slot. The SUs access the channel as follows. Each SU senses the band assigned to it for a duration of $\tau$ seconds, which is assumed to be a fraction of the slot duration, $T$. We assume that $\tau$ is chosen such that the probability of an erroneous secondary decision regarding primary activity is negligible. If the band is sensed to be free from primary activity, the SU, which is assigned to this band, transmits till the far end of the time slot. Note that the transmission time is $T-\tau$ not $T$, but it still transmits one full packet. This can be implemented by the terminal via adjusting its transmission rate, e.g., by using a signal constellation with more symbols or by increasing the channel coding rate or both. Note that by doing this, the probability of link outage increases. This is the price of transmission delay relative to the beginning of the time slot and it is exactly quantified at the end of this section.

For system $\mathcal{S}$, each band has at most one SU, and each SU is assigned only to one band. Let $\omega_{jk}$ denote the probability that $s_k$ is assigned to the $j${\it th} band. It is evident that we have two constraints. The first constraint is
\begin{equation}
\begin{split}
\sum_{k=1}^{\mathcal{M}_s}\omega_{jk}\le 1, \forall j\in\{1,\dots,\mathcal{M}_p\}
\end{split}
\label{constraints_omega1}
\end{equation}
The inequality holds to equality in case $\mathcal{M}_s \ge \mathcal{M}_p$. The second constraint is
\begin{equation}
\begin{split}
\sum_{j=1}^{\mathcal{M}_p}\omega_{jk}\le1, \forall k\in\{1,\dots,\mathcal{M}_s\}
\end{split}
\label{constraints_omega}
\end{equation}
where the inequality holds to equality in case $\mathcal{M}_p \ge \mathcal{M}_s$. Note that both constraints become equalities if and only if $\mathcal{M}_p=\mathcal{M}_s$.

If the number of SUs is greater than the available primary bands, and since our protocol does not allow multiple assignment of users to the same band, we can assume the presence of $\mathcal{M}_s\!-\!\mathcal{M}_p$ virtual bands with zero bandwidth. Thus, the service rate on any of these bands is exactly equal to zero. We define probability $q(m_1,m_2,\dots,m_{\mathcal{M}_s})$ as the probability that $s_1$ is assigned to the primary band $m_1$ and user $s_2$ is assigned to the primary band $m_2$ and so on, where $m_k\in\{1,2,\dots,\mathcal{M}_p\}$ for all $k=1,2,\dots,\mathcal{M}_s$ if $\mathcal{M}_p \ge \mathcal{M}_s$, and $m_k\in\{0,1,2,\dots,\mathcal{M}_p\}$ if $\mathcal{M}_s > \mathcal{M}_p$ with $m_k=0$ meaning that the SU is assigned to a virtual band. It is evident that the assignments are the permutation without repetition of choosing $\mathcal{M}_s$ elements out of $\mathcal{M}_p$ elements, if $\mathcal{M}_p \ge \mathcal{M}_s$, or choosing $\mathcal{M}_p$ elements out of $\mathcal{M}_s$ elements, if $\mathcal{M}_s \ge \mathcal{M}_p$. The total number of the assignments is given by
\begin{equation}
\mathbb{P}=\bigg[\frac{\mathcal{M}_p!}{(\mathcal{M}_p-\mathcal{M}_s)!}\bigg]^{\mathds{1}\left(\mathcal{M}_p\ge \mathcal{M}_s\right)}\bigg[\frac{\mathcal{M}_s!}{(\mathcal{M}_s-\mathcal{M}_p)!}\bigg]^{\mathds{1}\left(\mathcal{M}_s> \mathcal{M}_p\right)}
\end{equation}
where $\mathds{1}\left(.\right)$ is the indicator function and $r!$ denotes the factorial of $r$. It is clear that the summation over these probabilities satisfies the constraint
\begin{equation}
\sum_{(m_1,m_2,\dots,m_{\mathcal{M}_s})}q(m_1,m_2,\dots,m_{\mathcal{M}_s})= 1
\end{equation}
 The probability that band $B_j$ is free/available is the probability that the primary queue assigned to the band is empty, which is given by\footnote{This formula is followed from solving the Markov chain of the primary queue under the late-arrival model described in Section \ref{system_model} \cite{sadek}.}
\begin{equation}
\pi_{j}=1-\frac{\lambda_{p_j}}{\mu_{p_j}}
\label{eqn1}
\end{equation}
where $\mu_{p_j}$ is the mean service rate of $p_j$ and it is given by the complement of the outage event of the channel between the primary transmitter $p_j$ and its respective receiver under perfect sensing assumption.

We summarize MAC operation of system $\mathcal{S}$ as follows.
%\begin{itemize}
%  \item At the beginning of the time slot the PUs with nonempty queue transmit the packet at the head of their queues, and a band is assigned to one and only one SU.
%  \item The SUs sense the channel for $\tau$ seconds from the beginning of the time slot. A secondary transmitter with a nonempty queue transmits the packet at the head of its queue if the band is sensed to be free.
%  \item The terminals operate under a collision-channel model where concurrent transmissions are assumed to be lost. A feedback message from the respective receiver at the end of each time slot indicates to the transmitter the status of packet decoding. A correctly received packet is removed from the respective transmitter's queue.
%  \end{itemize}
  \begin{itemize}
  \item At the beginning of the time slot, the PUs with nonempty queue transmit the packet at the head of their queues, and a band is assigned to one and only one SU.
  \item The SUs sense the channel for $\tau$ seconds from the beginning of the time slot. A secondary transmitter with a nonempty queue transmits the packet at the head of its queue if the band is sensed to be free.
      \item  A feedback message from the respective receiver at the end of each time slot indicates the corresponding transmitter about the decodability status of the transmitted packet.
       \item If the respective destination decodes the packet successfully, it sends
back an acknowledgement (ACK), and the packet is removed from the system.
\item If the respective destination fails to decode the packet due to channel outage, it sends back a negative-acknowledgement (NACK), and the packet is retransmitted at the following time slot.
%\item A feedback message from the respective receiver at the end of each time slot indicates to the transmitter the status of packet decoding. A correctly received packet is removed from the respective transmitter's queue.
  \end{itemize}

% %%%%%%%%%%%%%%%%%%%%%%%%%%%
%%%%%%%%%%%%%%%%%%%
%%%%%%%%%%%%%%%%PHYSICAL LAYER ASSUMPTIONS
%%%%%%%%%%%%%%%%%%%%
%%%%%%%%%%%%%%%%%%%%%%%%
We adopt a flat fading channel model and assume that the channel gains remain constant over the duration of the time slot. We do not assume the availability of the CSI at the transmitting terminals.  Assuming that the number of bits in a packet is $b$, the transmission rate of the secondary transmitter $s_k$ is
\begin{equation}
r_{s_k}=\frac{b}{T-\tau}
\label{r_i}
\end{equation}
\noindent Outage occurs when the transmission rate exceeds the channel capacity \cite{sadek}
\begin{equation}
{\rm Pr}\bigg\{O_{i,s_k}\bigg\}=P_{{\rm out},is_k}={\rm Pr}\biggr\{r_{s_k} > W_i \log_{2}\left(1+\gamma_{s_k} \alpha_{is_k}\right)\biggr\}
\end{equation}
\noindent where $O_{i,s_k}$ is the event of channel outage when the $i${\it th} band is assigned to user $s_k$, $W_i$ is the bandwidth of the $i${\it th} band, $\gamma_{s_k}$ is the received signal-to-noise-ratio (SNR) at user $s_k$ receiver when the channel gain is equal to unity, and $\alpha_{is_k}$ is the channel gain when user $s_k$ is assigned the $i${\it th} band, which is exponentially distributed in the case of Rayleigh fading. The outage probability can be written as
\begin{equation}
P_{{\rm out},i{s_k}}={\rm Pr}\Biggr\{\alpha_{is_k}<\frac{2^{\frac{r_{s_k}}{W_i}}-1}{\gamma_{s_k}}\Biggr\}
\end{equation}
\noindent Assuming that the mean value of $\alpha_{i s_k}$ is $\sigma^2_{s_k}$, $P_{{\rm out},is_k}\!=\!1\!-\!\exp\bigg(-\frac{2^{\frac{r_{s_k}}{W_i}}-1}{\gamma_{s_k}\sigma^2_{s_k}}\bigg)$ for a Rayleigh fading channel. Let $\overline{P}_{{\rm out},is_k}=1-P_{{\rm out},is_k}$\footnote{Throughout the paper $\overline{z}=1-z$.}
 be the probability of the complement event $\overline{O}_{i,s_k}$. This probability of {\bf correct} packet reception is therefore given by
\begin{equation}
\overline{P}_{{\rm out},is_k}=\exp\bigg(-\frac{2^{\frac{b}{TW_i\left(1-\frac{\tau}{T}\right)}}-1}{\gamma_{s_k}\sigma^2_{s_k}}\bigg)
\label{corrprob}
\end{equation}
Note that the virtual bands are of unity outage probability because the available bandwidth is zero. The packet correct reception probability of user $p_i$ transmitting to its respective receiver is given by a similar formula as in Eqn. (\ref{corrprob}) with the respective primary parameters. Mathematically,
\begin{equation}
\overline{P}_{{\rm out},ip_i}=\exp\bigg(-\frac{2^{\frac{b}{TW_i}}-1}{\gamma_{p_i}\sigma^2_{p_i}}\bigg)
\label{corrprob2}
\end{equation}
%%%%%%%%%%%%%%%%%%%%%%%%%%%%%%%%%%%%%%%%%%%%%%%%%%%%%%%%%%%%%%%%%%%%%%%%%%5
%%% Stability Analysis of the System
%%%%%%%%%%%%%%%%%%%%%%%%%%%%%%%%%%%%%%%%%%%%%%%%%%%%%%%%%%%%%%%%%%%%%%%%%%%%
\section{Stability Analysis of System $\mathcal{S}$}\label{system_S}
A fundamental performance measure of a communication network is the stability of the queues. Stability can be defined rigorously as follows. For an irreducible and aperiodic Markov chain with
countable number of states, the chain is stable if and only if
there is a positive probability for every queue of being empty. Denote by $Q^{\left(t\right)}$ the length of queue $Q$ at the beginning of time slot $t$. Queue $Q$ is said to be stable if \cite{sadek}
%\vspace{-0.2 cm}
\begin{equation}\label{stabilityeqn}
    \lim_{x \rightarrow \infty  } \lim_{t \rightarrow \infty  } {\rm Pr}\{Q^{\left(t\right)}<x\}=1
\end{equation}

In a multiqueue system, the system is stable when {\bf all} queues are stable. We can apply Loynes' theorem to check the stability of a queue \cite{sadek}. This theorem states that if the arrival process and the service process of a queue are strictly stationary, and the average service rate is greater than the average arrival rate of the queue, then the queue is stable. If the average service rate is less than the average arrival rate, then the queue is unstable.

According to the adopted arrival model described in the system model, the queue $Q_\nu$ evolves as follows:
\begin{equation}
    Q_\nu^{t+1}=(Q_\nu^t-D^t_\nu)^+ +A^t_\nu
\end{equation}
where $D_\nu^t$ is the number of departures of queue $Q_\nu$ at time slot $t$, $A_\nu^t$ is the number of arrivals to $Q_\nu$ at time slot $t$, and $(\zeta)^+$ denotes $\max\{\zeta,0\}$.

 The $j${\it th} primary queue is stable when $\lambda_{p_j} < \mu_{p_j}$. Let $\mu_{s_k}$ be the mean service rate of the queue of user $s_k$, $Q_{s_k}$. Recall that the probability of assigning user $s_k$ to band $m_k$ in a certain time slot is $\omega_{{m_k}k}$. The relationship among $\omega_{{m_k}k}$ and the $q$'s can be stated as follows.
\begin{equation}
\begin{split}
\omega_{{m_k}k}=\sum_{{\sim}{m_k}}q(m_1,m_2,\dots,m_{\mathcal{M}_s}), \forall k\in\{1,\dots,\mathcal{M}_s\}
\end{split}
\label{constraints_omega}
\end{equation}
\noindent where the sum is over all indices except $m_k$. The mean service rate of the $j${\it th} PU is given by
\begin{equation}
\mu_{p_j}= \overline{P}_{{\rm out},{j{p_j}}}, \,\ \forall \ j=1,2,\dots,\mathcal{M}_p
\label{eqn4}
\end{equation}

 A packet at the head of user $s_k$ queue is served if the band in which $s_k$ is assigned to is available and the channel to its respective receiver is not in outage. Let $\mathcal{P}_{m_kk}=\pi_{m_k} \overline{P}_{{\rm out},{{m_k}{s_k}}}$, the mean service rate of user $s_k$ is given by:
\begin{equation}
\mu_{s_k}= \sum_{(m_1,m_2,\dots,m_{\mathcal{M}_s})}  q(m_1,m_2,\dots,m_{\mathcal{M}_s}) \ \mathcal{P}_{m_kk}
\label{eqn6}
\end{equation}
where the sum is over all possible assignments of $(m_1,m_2,\dots,m_{\mathcal{M}_s})$. It is worth noting that during the fraction of operation time, $q(m_1,m_2,\dots,m_{\mathcal{M}_s})$, the average service rate of user $s_k$ is $\mathcal{P}_{m_kk}$. Using the relationship among $\omega$'s and $q$'s in (\ref{constraints_omega}), we can rewrite (\ref{eqn6}) as follows:
\begin{equation}
\mu_{s_k}= \sum_{m_k=1}^{\mathcal{M}_p} \omega_{{m_k}k} \ \mathcal{P}_{m_kk}
\label{omeg_for}
\end{equation}
Expression (\ref{omeg_for}) is interpreted as follows. The $k${\it th} SU is served if it is assigned to the primary band $m_k$, which occurs with probability $\omega_{{m_k}k}$, while this band is free/available and the associated channel to the $k${\it th} SU respective receiver is not in outage. The sum in (\ref{omeg_for}) is over the $\mathcal{M}_p$ primary bands.

The stability region is characterized by the closure of rates $(\lambda_{s_1},\lambda_{s_2},\dots,\lambda_{\mathcal{M}_s})$. One method to characterize this closure is to solve a constrained
optimization problem to find the maximum feasible $\lambda_{s_k}$ corresponding
to each feasible $\lambda_{s_\ell}$, $\ell \neq k$, with all the system queues being stable \cite{sadek}. Specifically, for fixed $\lambda_{s_\ell}$, for all  $\ell \neq k$, the maximum stable throughput region is obtained via solving the following optimization problem:

\begin{equation}
\begin{split}
     &  \underset{q(m_1,m_2,\dots,m_{\mathcal{M}_s})\!\ge0}{\max.} \lambda_{s_k}\!=\! \underset{(m_1,m_2,\!\dots\!,m_{\mathcal{M}_s}\!)}{\sum}  \bigg[ q(\!m_1,m_2,\!\dots,m_{\mathcal{M}_s}\!)\ \\& \,\,\,\,\,\,\,\,\,\,\,\,\,\,\,\,\,\,\,\,\,\,\,\,\,\,\,\,\,\,\,\,\,\,\,\,\,\,\,\,\,\,\,\,\,\,\,\,\,\,\,\,\,\,\,\,\,\,\,\,\,\,\,\,\,\,\,\,\,\,\,\,\,\,\,\,\,\,\,\,\,\,\,\,\,\,\,\,\,\,\,\,\,\,\,\,\,\,\,\,\,\,\,\,\,\,\,\,\,\,\,\,\,\,\,\,\,\,\,\,\,\ \times \mathcal{P}_{m_kk}\Bigg] \\& \,\, {\rm s.t.} \sum_{(m_1,m_2,\dots,m_{\mathcal{M}_s})} q(m_1,m_2,\dots,m_{\mathcal{M}_s})=1,\ \\& \,\ \lambda_{s_\ell}\!\le \! \sum_{(m_1,m_2,\dots,m_{\mathcal{M}_s})}   q(m_1,m_2,\dots,m_{\mathcal{M}_s})\ \mathcal{P}_{m_{\ell} \ell}, \forall \ell \neq k
    \end{split}
    \end{equation}
    The optimization problem is a linear program and can be solved using any standard linear programming technique. In order to decrease the total number of optimization variables, we use an equivalent optimization problem, i.e., in terms of $\omega$'s. Defining matrix $\Omega$ such that its $jk$ element is $\omega_{jk}$ and using (\ref{omeg_for}), the optimization problem can be rewritten as follows:
   \begin{equation}
\begin{split}
     &  \underset{\Omega}{\max.}\,\,\ \lambda_{s_k}\!=\! \sum_{m_k=1}^{\mathcal{M}_p} \omega_{{m_k}k} \ \mathcal{P}_{m_kk}\\& \,\, {\rm s.t.} \ 0\le \omega_{{m_h}h}\!\le\!1 \ \forall (m_h,h) ,\ \sum_{m_h}^{\mathcal{M}_p} \omega_{{m_h}h}\le1 \ \forall h,\ \\&\,\,\,\,\,\,\,\,\,\,\,\,\,\  \sum_{h=1}^{\mathcal{M}_s} \omega_{{m_h}h}\le 1 \ \forall m_h, \ \lambda_{s_\ell}\!\le\! \sum_{m_\ell=1}^{\mathcal{M}_p} \omega_{{m_\ell}\ell}\ \mathcal{P}_{m_\ell \ell } \forall  \ell \neq k
    \end{split}
    \end{equation}
  \noindent where $h,\ell\in \{1,2,\dots,\mathcal{M}_s\}$ and $m_h,m_\ell\in \{1,2,\dots,\mathcal{M}_p\}$. The optimization problem is still a linear program, which can be solved efficiently. It has a total number of variables $\mathcal{M}_s \times \mathcal{M}_p$ which is much less than the total number of variables of the original problem, i.e., $\mathcal{M}_s \times \mathcal{M}_p\ll \mathbb{P}$. For the operation of the system, we can obtain $q$'s from $\Omega$'s using Birkhoff algorithm (see, for example, \cite{li2001enhanced} and references therein), which gives the $q$'s or the fraction of time in which a certain users configuration is used. The Birkhoff algorithm is applied on square doubly stochastic matrices.\footnote{A doubly stochastic matrix (also called bistochastic), is a matrix $A=(a_{jk})$ of nonnegative real numbers and each of its rows and columns sums to unity, i.e., $\underset{j}{\sum} a_{jk}=\underset{k}{\sum} a_{jk}=1$.} Therefore, if $\mathcal{M}_s>\mathcal{M}_p$, we can assume there are virtual bands of zero bandwidth to which $\mathcal{M}_s-\mathcal{M}_p$ users are assigned. If $\mathcal{M}_p>\mathcal{M}_s$, we assume that there are virtual SUs with zero-arrival rate and unity outage probability.
%\cite{chang2000birkhoff,li2001enhanced,chang2002load,peng2006quick}

    Now we move our attention to the case of two SUs and two PUs (two bands) to obtain some insights and analytical results for the stability region. Since $\mathcal{M}_s=\mathcal{M}_p=2$ and from (\ref{constraints_omega1}) and (\ref{constraints_omega}), $\omega_{12}\!=\!\omega_{21}$\footnote{Since the SUs are assigned different bands each slot, the probability of assigning user $s_1$ to band $2$ is equal to the probability of assigning user $s_2$ to band $1$.}. The stability region is characterized by the closure of rate pairs $(\lambda_{s_1},\lambda_{s_2})$. The optimization problem is stated as:
    \begin{equation}
\begin{split}
     & \underset{\epsilon}{\max.} \,\,\,\,\,\,\,\,\,\,\ \epsilon \mathcal{P}_{12}+ \big(1-\epsilon\big)\mathcal{P}_{22} \\& \,\,\,\ {\rm s.t.}  \,\,\,\,\ \lambda_{s_1}\!\le \!  \epsilon \mathcal{P}_{21} \!+\! \big(1\!-\!\epsilon\big) \mathcal{P}_{11},\,\ 0\le \epsilon\le 1
     \label{opt_two}
    \end{split}
    \end{equation}
    where $\epsilon\!=\!\omega_{12}\!=\!\omega_{21}$ is the probability that user $s_2$ is assigned to band $1$ (or user $s_1$ is assigned to band $2$). The optimization problem can be rearranged as follows
    \begin{equation}
    \label{opt1001}
\begin{split}
     & \underset{\epsilon}{\max.} \,\,\,\,\,\,\,\ \epsilon \big(\mathcal{P}_{12}-\mathcal{P}_{22}\big)\\ &{\rm s.t.}  \,\,\,\,\,\,\,\ \lambda_{s_1}-\mathcal{P}_{11}\!\le \! \epsilon \big(\mathcal{P}_{21}- \mathcal{P}_{11}\big),\,\ 0\le \epsilon\le 1
    \end{split}
    \end{equation}
      The optimal $\epsilon$ depends on the values of $\mathcal{P}_{jk}$ for all $j,k\in\{1,2\}$. Specifically,
    \begin{itemize}
      \item  If $\mathcal{P}_{21}< \mathcal{P}_{11}$ and $\lambda_{s_1}>\mathcal{P}_{11}$; or $\mathcal{P}_{21}> \mathcal{P}_{11}$ and $\lambda_{s_1} > \mathcal{P}_{21}$, the problem is \textbf{infeasible}.
   \item  If $\mathcal{P}_{12}>\mathcal{P}_{22}$, $\mathcal{P}_{21}\ge \mathcal{P}_{11}$, and $\lambda_{s_1} \le \mathcal{P}_{21}$, the optimal value is $\epsilon^*=1$.

   \item  If $\mathcal{P}_{12}>\mathcal{P}_{22}$, $\mathcal{P}_{21}< \mathcal{P}_{11}$ and $\lambda_{s_1}-\mathcal{P}_{11}<0$, the optimal value is $\epsilon^*=\min\big(\frac{ \lambda_{s_1}-\mathcal{P}_{11}}{ \mathcal{P}_{21}- \mathcal{P}_{11}},1\big)$.
      \item   If $\mathcal{P}_{12}<\mathcal{P}_{22}$ and $\mathcal{P}_{21}> \mathcal{P}_{11}$, the optimal value is $\epsilon^*=\max\big(\frac{ \lambda_{s_1}-\mathcal{P}_{11}}{ \mathcal{P}_{21}- \mathcal{P}_{11}},0\big)$.
    \item  If $\mathcal{P}_{12}<\mathcal{P}_{22}$, $\mathcal{P}_{21}< \mathcal{P}_{11}$ and $\lambda_{s_1}\le\mathcal{P}_{11}$, the optimal value is $\epsilon^*=0$.
        %%% $\frac{\lambda_{s_1}-\mathcal{P}_{11}}{\mathcal{P}_{11}-\mathcal{P}_{21}}\!\ge \! \epsilon $
         \item If $\mathcal{P}_{12}=\mathcal{P}_{22}$, the optimization problem becomes a feasibility problem.
         \end{itemize}

      The stability region is given by
         \begin{equation}\label{2233344}
   \mathcal{R}(\mathcal{S})=\bigg\{(\lambda_{s_1},\lambda_{s_2}):\lambda_{s_2} <  \epsilon^* \mathcal{P}_{12}+ \bigg(1-\epsilon^*\bigg)\mathcal{P}_{22}\bigg\}
\end{equation}
%It should be noted that $\epsilon^*$ guarantees the stability of $Q_{s_1}$ and achieves the maximum throughput of user $s_2$ [achieves the maximum of objective function and satisfies the optimization problem's constraints in (\ref{opt_two})].
%\begin{proposition}\label{pro3}
%The stability region of system $\mathcal{S}$ for $\mathcal{M}_s=\mathcal{M}_p=2$ is a \textbf{convex} set. In particular, it is a polyhedron.
%\end{proposition}
%\begin{proof}
%From (\ref{2233344}), the stability region's boundary, $\lambda_{s_2} =  \epsilon^* \mathcal{P}_{12}+ \big(1-\epsilon^*\big)\mathcal{P}_{22}$, is a convex function on $\epsilon^*$ (affine function). On the other hand, $\epsilon^*$ is a convex function\footnote{The intersection of two affine sets.} of $\lambda_{s_1}$ (see solution of problem (\ref{opt1001})) and, hence, the set of rate pairs $(\lambda_{s_1},\lambda_{s_2})$ is convex. Note that since $\mathcal{R}(\mathcal{S})$ is convex, this implies that for any given two stable rate pairs, the line segment connecting them is also in the set and, hence, is composed of stable rate pairs.
%\end{proof}

    %%%%%%%%%%%%%%%5
    %%%%%%%%%%%%%%
    %%%%%%%%%%%%%%
    %%%% Probabalistic assignment
%    \vspace{-0.3 cm}
    \section{Comparison Baseline Systems}\label{comparisonsection}
%    \vspace{-0.28 cm}
    In this section, we consider two systems for comparison with the proposed system. The first system, denoted by $\mathcal{\hat{S}}$, needs less coordination and cooperation between SUs. Each SU chooses a band randomly at the beginning of the time slot. The probability that user $s_k$ chooses band $i$ is $\Gamma_{ik}$. It is clear that these probabilities satisfy the constraint
\begin{equation}
\sum_{i=1}^{\mathcal{M}_p}\Gamma_{ik}\le 1,\, \forall k \in\{1,\dots,\mathcal{M}_s\}
\end{equation}

It is possible in system $\mathcal{\hat{S}}$ that a band is left unassigned or that several SUs are assigned to the same band. In this system, packet loss is due to either packets collision, when two or more SUs with nonempty queues select the same primary band; or channels outage. The total number of assignment of SUs to bands is given by
\begin{equation}
\mathbb{C}=\mathcal{M}_p^{\mathcal{M}_s}
\end{equation}

System $\mathcal{\hat{S}}$ is less complex than system $\mathcal{S}$ due to the lack of need for strict coordination between the secondary terminals, which is required in $\mathcal{S}$ where one and only one user is given a specific band. Nevertheless, the complexity of obtaining the optimal assignments probability is much higher than system $\mathcal{S}$ because the optimization problem of system $\mathcal{\hat{S}}$ is nonconvex and the total number of optimization parameters is $\mathcal{M}_p^{\mathcal{M}_s}> \mathcal{M}_p\times{\mathcal{M}_s}$ when $\mathcal{M}_p\ge2$ and $\mathcal{M}_p^{\mathcal{M}_s}\gg \mathcal{M}_p\times{\mathcal{M}_s}$ for high $\mathcal{M}_p$ or $\mathcal{M}_s$.

The mean service rate of the $j${\it th} PU is similar in systems $\mathcal{S}$ and $\mathcal{\hat{S}}$. We investigate now the service rate for the SUs. User $s_k$, when assigned to band $B_i$, succeeds in its transmission with probability $\overline{P}_{{\rm out},is_k}$ if the PU operating on $B_i$ has no packets to send and if any secondary terminal assigned to the same band has an empty queue. The mean service rate of user $s_k$ is thus given by
\begin{equation}\label{genwoc}
\begin{split}
\mu_{s_k}=&\sum_{m_1=1}^{\mathcal{M}_p}\sum_{m_2=1}^{\mathcal{M}_p}...\sum_{m_{\mathcal{M}_s}=1}^{\mathcal{M}_p}\Biggr[\Gamma_{m_11}\Gamma_{m_22}...\Gamma_{m_{\mathcal{M}_s}\mathcal{M}_s}\ \mathcal{P}_{m_kk}\\ & \,\,\,\,\,\,\,\,\,\,\,\,\,\,\,\,\,\,\,\,\,\,\,\,\,\,\,\,\,\,\,\,\,\,\,\,\,\,\,\,\,\,\,\,\,\,\,\,\,\,\,\,\,\,\,\,\,\ \times \,{\rm Pr}\Bigg\{ \bigcap^{\mathcal{M}_s}_{ \substack{v=1 \\ v \neq k \\ m_v = m_k}} Q_{s_v} =0 \Bigg\} \Biggr]
\end{split}
\end{equation}
where the sums in (\ref{genwoc}) are over the possible assignments of
each SU.
\begin{proposition}\label{pro2}
For any network with $\mathcal{M}_s$ SUs and $\mathcal{M}_p$ orthogonal primary bands, the stability region of system $\mathcal{S}$, $\mathcal{R}(\mathcal{S})$, contains
that of $\hat{\mathcal{S}}$, $\mathcal{R}(\mathcal{\hat{S}})$. That is, $\mathcal{R}(\mathcal{\hat{S}}) \subseteq  \mathcal{R}(\mathcal{S})$.
\end{proposition}
\begin{proof}
We investigate the system with $\mathcal{M}_p \geq \mathcal{M}_s$ first. Assume the same pattern of queue occupancy in both systems. The mean service rate of user $s_k$ with a nonempty queue is
\begin{equation}
\begin{split}
\mu^{\left(\hat{\mathcal{S}}\right)}_{s_k}&=  \sum_{m_k=1}^{\mathcal{M}_p} \mathcal{P}_{m_kk}\ \Gamma_{{{m_k}k}} \prod_{\substack{{v\in \mathcal{N}}\\{v\ne k}}}(1-\Gamma_{{m_k}v})
\end{split}
\label{hatSe}
\end{equation}
\noindent where $\mathcal{N}$ is the set of SUs with nonempty queues. Note that we use the superscript $\hat{\mathcal{S}}$ to make it clear that expression (\ref{hatSe}) is for system $\hat{\mathcal{S}}$. On the other hand, for system $\mathcal{S}$,
\begin{equation}
\begin{split}
\mu^{\left(\mathcal{S}\right)}_{s_k}&= \sum_{(m_1,m_2,\dots,m_{\mathcal{M}_s})} q(m_1,m_2,\dots,m_{\mathcal{M}_s})\ \mathcal{P}_{m_kk} \\&=\sum_{m_k=1}^{\mathcal{M}_p}  \omega_{m_k k}\ \mathcal{P}_{m_kk}
\end{split}
\label{Se}
\end{equation}
Subtracting (\ref{hatSe}) from (\ref{Se}), we get
\begin{equation}
\begin{split}
\mu^{\left(\mathcal{S}\right)}_{s_k}-\mu^{\left(\hat{\mathcal{S}}\right)}_{s_k}
% \sum_{m_k=1}^{\mathcal{M}_p}  \omega_{m_k k} \ \mathcal{P}_{m_kk}\\& \,\,\,\,\,\,\,\,\ -\sum_{m_k=1}^{\mathcal{M}_p} \mathcal{P}_{m_kk} \ \Gamma_{{{m_k}k}} \prod_{\substack{{v \in \mathcal{N}}\\{v\ne k}}}(1-\Gamma_{{m_k}v})
%\\
&= \sum_{m_k=1}^{\mathcal{M}_p}\mathcal{P}_{m_kk} \bigg(\omega_{m_k k} - \Gamma_{{{m_k}k}} \prod_{\substack{{v \in \mathcal{N}}\\{v\ne k}}}(1-\Gamma_{{m_k}v})\bigg)
\end{split}
\end{equation}
%\color{blue}
%\begin{equation}
%\begin{split}
%\mu^{\left(\mathcal{S}\right)}_{s_k}-\mu^{\left(\hat{\mathcal{S}}\right)}_{s_k}&= \sum_{m_k=1}^{\mathcal{M}_p}\mathcal{P}_{m_kk} \bigg(\omega_{m_k k} - \Gamma_{{{m_k}k}} \prod_{\substack{{v \in \mathcal{N}}\\{v\ne k}}}(1-\Gamma_{{m_k}v})\bigg)\\& \ge \sum_{m_k=1}^{\mathcal{M}_p}\mathcal{P}_{m_kk} \bigg(\omega_{m_k k} - \Gamma_{{{m_k}k}}\bigg)
%\label{bl}
%\end{split}
%\end{equation}
%
%Recall that $\sum_{m_k=1}^{\mathcal{M}_p} \omega_{m_k k}=1$ and  $\sum_{m_k=1}^{\mathcal{M}_p}\Gamma_{{{m_k}k}}\le 1$. Setting $\omega_{m_k k}= \Gamma_{{{m_k}k}}$ provides mean service rates for system $\mathcal{S}$ higher than system $\hat{\mathcal{S}}$ [Eqn. (\ref{bl}) becomes $\mu^{\left(\mathcal{S}\right)}_{s_k}-\mu^{\left(\hat{\mathcal{S}}\right)}_{s_k} \ge 0$]. However, optimizing over $\omega_{m_k k}$ can achieve much better performance for system $\mathcal{S}$.
%%\begin{equation}
%%\begin{split}
%%\mu^{\left(\mathcal{S}\right)}_{s_k}-\mu^{\left(\hat{\mathcal{S}}\right)}_{s_k}&  \ge 0
%%\end{split}
%%\end{equation}
%\color{black}
Note that $\sum_{m_k=1}^{\mathcal{M}_p}\Gamma_{{{m_k}k}} \prod_{\substack{{v \in \mathcal{N}}\\{v\ne k}}}(1-\Gamma_{{m_k}v})$ represents the probability of one user being assigned a certain band with all other users with nonempty queues being assigned to another band. This configuration is a subset of all possible users' assignments which additionally include a situation with two or more users with nonempty queues assigned to a band and the rest of users assigned to another band. This means that the sum given by $\sum_{m_k=1}^{\mathcal{M}_p}\Gamma_{{{m_k}k}} \prod_{\substack{{v \in \mathcal{N}}\\{v\ne k}}}(1-\Gamma_{{m_k}v})$ is less than or equal to $1$. Since $\sum_{m_k=1}^{\mathcal{M}_p}\omega_{m_k k}=1$, we can always find $\omega_{m_k k} \ge \Gamma_{{{m_k}k}} \prod_{\substack{{v \in \mathcal{N}}\\{v\ne k}}}(1-\Gamma_{{m_k}v})$.
%We can always find $\omega_{m_k k} \ge \Gamma_{{{m_k}k}} \prod_{\substack{{v \in \mathcal{N}}\\{v\ne k}}}(1-\Gamma_{{m_k}v})$ because
%\begin{equation}
%\sum_{m_k=1}^{\mathcal{M}_p}\omega_{m_k k}=1\  \hbox{while}\ \sum_{m_k=1}^{\mathcal{M}_p}\Gamma_{{{m_k}k}} \prod_{\substack{{v \in \mathcal{N}}\\{v\ne k}}}(1-\Gamma_{{m_k}v}) \le1
%\end{equation}

 Now if $\mathcal{M}_p\le \mathcal{M}_s$, this case can be seen as a system with $\mathcal{M}_p=\mathcal{M}_s$ with $\mathcal{M}_s- \mathcal{M}_p$ zero-bandwidth bands. Thus, we can infer that $\mathcal{R}(\mathcal{S})$ contains $\mathcal{R}(\hat{\mathcal{S}})$ in all cases. This completes the proof.
\end{proof}
    %%%%%%%%%%%%%%%5
    %%%%%%%%%%%%%%
    %%%%%%%%%%%%%%
    %%%% Deterministic assignment

The second system that we compare with is the deterministic (fixed) assignment system in which the SUs are deterministically assigned to the primary bands. That is, each SU is assigned to one of the primary bands for all time. Hence, this system requires that $\mathcal{M}_p\!\ge\! \mathcal{M}_s$.

\begin{proposition}\label{pro1}
For $\mathcal{M}_s$ SUs and $\mathcal{M}_p\ge \mathcal{M}_s$ bands, the stability regions of system $\mathcal{S}$ and $\mathcal{\hat{S}}$ contain that of a fixed assignment.
\end{proposition}
\begin{proof}
The fixed assignment system is a special case of system $\mathcal{S}$ corresponding to the case where the probability $q(m_1,m_2,m_3,\dots,m_{\mathcal{M}_s})$ of the assignment is unity and all the other probabilities are zero. In addition, the fixed assignment system is a special case of system $\hat{\mathcal{S}}$ with  $\Gamma_{ik}$ set to unity when band $i$ is allocated to $s_k$ and zero otherwise. Therefore, both systems $\mathcal{S}$ and $\hat{\mathcal{S}}$ are superior to a fixed assignment.
\end{proof}
%   \begin{equation}\label{analy33}
%   \begin{split}
%   \mathcal{R}(\mathcal{S}^{\left(D\right)})&=\bigg\{(\lambda_{s_1},\lambda_{s_2},\dots,\lambda_{s_{\mathcal{M}_s}}):\lambda_{s_j} < \pi^{\left({\rm SE}\right)}_{m_j}   \overline{P}_{{\rm out},{m_j{s_j}}}\bigg(1-P_{\rm FA}^{\left({s_j}{p_{m_j}}\right)}\bigg), \\& \,\,\,\,\,\,\,\,\,\,\,\,\,\,\,\,\,\,\,\,\,\ 0<\lambda_{s_\ell}<  \pi^{\left({\rm SE}\right)}_{m_\ell}   \overline{P}_{{\rm out} ,{m_\ell{s_\ell}}} \bigg(1-P_{\rm FA}^{\left({s_\ell}{p_{m_\ell}}\right)}\bigg),\,\,\ \forall \ell=1,2,\dots,\mathcal{M}_s \,\ \hbox{and}\,\,\ \ell \ne j \bigg\}
%   \end{split}
%\end{equation}

%%%%%%%%%%%%%%%%%%%%%%%%%%%%%%%%%%%5
%%%%%%%%%%%%%%%%%%%%%%%%%%%%%%%%%%%
%%%%%%%%%%%%%%%%%%%%%%%%%%%%%%%%%%%5
%%%%%%%%%%%%%%%%%%%%%%%%%%%%%%%%%%%
%% NUMERICAL RESULTS
%% NUMERICAL RESULTS
%%%%%%%%%%%%%%%%%%%%%%%%%%%%%%%%%%%%
%%%%%%%%%%%%%%%%%%%%%%%%%%%%%%%%%%%%%
%%%%%%%%%%%%%%%%%%%%%%%%%%%%%%%%%
\begin{table*}
\begin{center}
\caption{The complement of channels outage for the secondary nodes and the bands availability of the primary bands Used to generate Figs. 2 and 3.}
\label{table2}
\begin{tabular}{  |c |c |c|c|c|c|  }
    \hline\hline
  \hbox{User $s_1$}& \hbox{User $s_2$}&\hbox{User $s_3$}& \hbox{User $s_4$} & \hbox{Band Availability}\\[5pt]\hline
    $\overline{P}_{{\rm out},1{s_1}}=0.6$ &$\overline{P}_{{\rm out},1{s_2}}=0.7$& $\overline{P}_{{\rm out},1{s_3}}=0.6$ &$\overline{P}_{{\rm out},1{s_4}}=0.7$&$\pi_1=1-\frac{\lambda_{p_1}}{\overline{P}_{{\rm out},1p_1}}=0.45$ \\
      $\overline{P}_{{\rm out},2{s_1}}=0.8$ &$\overline{P}_{{\rm out},2{s_2}}=0.6$ &$\overline{P}_{{\rm out},2{s_3}}=0.8$ &$\overline{P}_{{\rm out},2{s_4}}=0.5$ &$\pi_2=1-\frac{\lambda_{p_2}}{\overline{P}_{{\rm out},2p_2}}=0.2$\\
     $ \overline{P}_{{\rm out},3{s_1}}=0.7$ &$\overline{P}_{{\rm out},3{s_2}}=0.8 $ &  $ \overline{P}_{{\rm out},3{s_3}}=0.7$ &$\overline{P}_{{\rm out},3{s_4}}=0.6 $ &$\pi_3=1-\frac{\lambda_{p_3}}{\overline{P}_{{\rm out},3p_3}}=0.6$\\
     $\overline{P}_{{\rm out},4{s_1}}=0.85 $ &$ \overline{P}_{{\rm out},4{s_2}}=0.9 $& $\overline{P}_{{\rm out},4{s_3}}=0.5 $ &$ \overline{P}_{{\rm out},4{s_4}}=0.95 $&$\pi_4=1-\frac{\lambda_{p_4}}{\overline{P}_{{\rm out},4p_4}}=0.4$ \\[5pt]\hline

\end{tabular}
\end{center}
\end{table*}

\section{Numerical Results}\label{numerical}
We provide here some numerical results for the optimization problems presented in this paper. Let $d(m_1,m_2)$ denote the fixed allocation of user $s_1$ to band $m_1$ and user $s_2$ to band $m_2$ in a system with $\mathcal{M}_s\!=\!\mathcal{M}_p\!=\!2$. Fig.\ \ref{fig:subfig1} provides a comparison between the stability regions of systems $\mathcal{S}$, $\hat{\mathcal{S}}$, $d(1,2)$ and $d(2,1)$. The parameters used to generate the figure are: $\overline{P}_{{\rm out},2{s_1}}\!=\!0.8$, $\overline{P}_{{\rm out},2{s_2}}\!=\!0.9$, $\overline{P}_{{\rm out},1{s_1}}\!=\!0.7$, $\overline{P}_{{\rm out},1{s_2}}\!=\!0.85$, and the bands availability are $\pi_{1}\!=\!1\!-\!\frac{\lambda_{p_1}}{\overline{P}_{{\rm out},1p_1}}\!=\! 0.25$ and $\pi_{2}~=\!1\!-\!\frac{\lambda_{p_2}}{\overline{P}_{{\rm out},2p_2}}\!=\!0.875$. From the figure, the advantage of systems $\mathcal{S}$ and $\hat{\mathcal{S}}$ over the deterministic assignment is noted. Also, the advantage of system $\mathcal{S}$ over all the considered systems is noted.

Fig. \ref{fig:subfig2} shows the stability region of system $\mathcal{S}$ in case of $\mathcal{M}_s\!=\!\mathcal{M}_p\!=\!4$. The figure reveals the impact of increasing the mean arrival rate of users $s_3$ and $s_4$ on the stability region of users $s_1$ and $s_2$. As shown in the figure, the increase in the mean arrival rates of users $s_3$ and $s_4$ reduces the stability region of users $s_1$ and $s_2$. The parameters used to generate the figure are depicted in the figure's caption and Table \ref{table2}. Fig. \ref{fig:subfig3} presents the optimal assignment probabilities for system $\mathcal{S}$ for the given parameters in the figure's caption. The parameters used to generate the figure are: $\mathcal{M}_s\!=\!\mathcal{M}_p\!=\!3$, $\lambda_{s_3}\!=\!\lambda_{s_4}\!=\!0.35$ packets per time slot and the first three rows and columns of users $s_1$, $s_2$ and $s_3$ in Table \ref{table2}. It can be noted that as the mean arrival rate of the second user, $s_2$, increases, $q^*(1,3,2)$ and $q^*(2,3,1)$ increase as well, which denote the probability that user $s_2$ is allocated to the third band. This is because the third band provides the highest $\mathcal{P}_{jk}$ for user $s_2$, i.e., $\mathcal{P}_{32}>\mathcal{P}_{j2}$ for $j=1,2$, and user $s_2$ needs to increase its service rate to maintain its queue stability. Similarly, as the mean service rate of user $s_1$ increases, the probabilities $q^*(3,2,1)$ and $q^*(3,1,2)$ increase for the same reason mentioned before for user $s_2$.

%comparison_hat_S_and_S_perfect_and_deterministic_3
\begin{figure}
	\centering
		\includegraphics[scale=0.5]{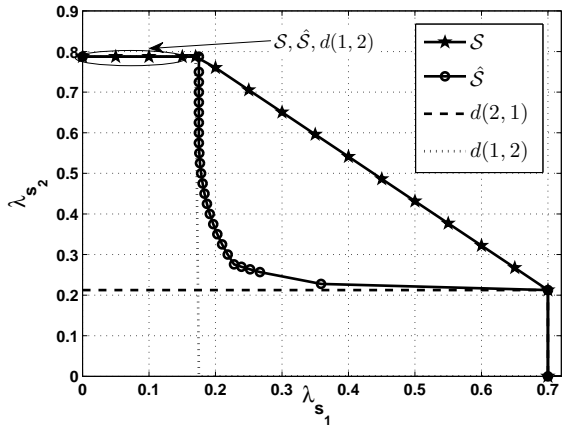}
	\caption{Stability regions of the considered systems.}
	\label{fig:subfig1}
\end{figure}
\begin{figure}
	\centering
		\includegraphics[scale=0.5]{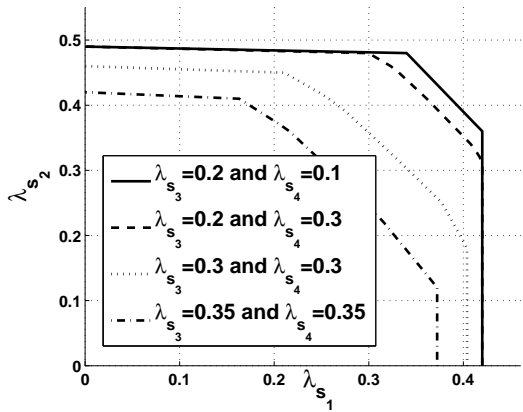}
	\caption{Stability region of system $\mathcal{S}$. The parameters used to generate the figure are: $\mathcal{M}_s\!=\!\mathcal{M}_p\!=\!4$ and Table \ref{table2}.}
	\label{fig:subfig2}
\end{figure}
%\begin{figure}[h]
% \caption[]
% \centering
% \subfigure[Stability region of the proposed systems and the possible stability region of a fixed (deterministic) assignment system.]{
%  \includegraphics[scale=0.5]{comparison_hat_S_and_S_perfect_and_deterministic_3}
%   \label{fig:subfig1}
%   }\hfill
% \subfigure[Stability region of system $\mathcal{S}$. The parameters used to generate the figure are: $\mathcal{M}_s\!=\!\mathcal{M}_p\!=\!4$ and Table \ref{table2}.]{
%  \includegraphics[scale=0.5]{S_four_users_4_bands_different_arrival_rates_for_users_3_4_DIFFERENT_OUTAGE.eps}
%   \label{fig:subfig2}
%  }
%  \bigskip
% % \subfigure[Stability region of the system $\mathcal{S}$. The parameters used to generate the figure are $\mathcal{M}_s$=4 and parameters in Table \ref{table5}]{
%%  \includegraphics[scale=0.5]{S_four_users_4_bands_6bands.eps}
%%   \label{fig:subfig3}
%%   }
%% \label{fig:subfigureExample}
%{%
% % \subref{},
%%  \subref{} \subref{}
%}
%\end{figure}
\begin{figure}
	\centering
		\includegraphics[scale=0.5]{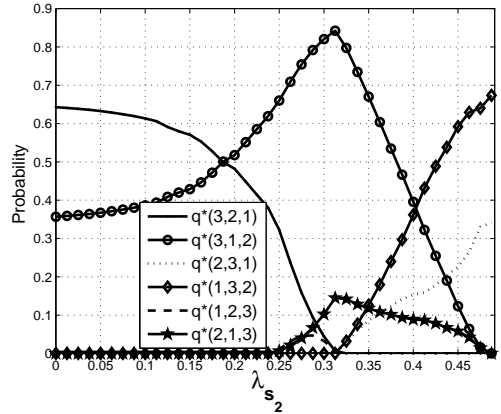}
	\caption{The optimal SUs' allocation probabilities for system $\mathcal{S}$ in case of $\mathcal{M}_s\!=\!\mathcal{M}_p\!=\!3$. The parameters used to generate the figure are $\lambda_{s_3}\!=\!\lambda_{s_4}\!=\!0.35$ packets per time slot and the first three rows and the columns of users $s_1$, $s_2$ and $s_3$ in Table \ref{table2}.}
	\label{fig:subfig3}
\end{figure}

%We have proposed a band allocation scheme for cognitive radio users. The cognitive radio users are allocated to bands based on their queues stability requirements. We have proved the advantage of the proposed scheme over some well-known schemes. Future research can be directed at systems with multipacket reception capabilities.
\section{conclusions}\label{conc}
We have proposed a band allocation scheme for buffered cognitive radio users in presence of orthogonal licensed primary bands each of which assigned to a PU.
The cognitive radio users are allocated to bands based on their queue stability requirements.
 We have proved the advantage of the proposed scheme over some well-known schemes.
 Future research for system $\mathcal{S}$ can be directed at one of the following points. 1) Considering systems with multiple assignment within one slot. More specifically, the assignment of users happens multiple time per slot to satisfy all users. The knowledge of the transmit CSI can enhance the system performance and allow bands exchange among users; 2) allowing priority among SUs such that multiple users can be assigned to the same band with different priority in band accessing. The priority of transmission can be established by making the lower priority user sense the higher priority user activity for certain time duration within the slot; or 3) another possible extension is to study the impact of sensing errors on the system's performance. For system $\hat{\mathcal{S}}$, the extension can be directed in terms of 1) adding multipacket reception to the receiving node; or 2) allowing band selection at different time instants per slot followed by sensing duration to avoid perturbing the current transmission.
\bibliographystyle{IEEEtran}
\bibliography{IEEEabrv,bandsbib}
\balance
%\begin{figure}
% \caption[{}]
% \centering
%  \subfigure[The optimal SUs' allocation probabilities for system $\mathcal{S}$ in case of $\mathcal{M}_s\!=\!\mathcal{M}_p\!=\!3$. The parameters used to generate the figure are $\lambda_{s_3}\!=\!\lambda_{s_4}\!=\!0.35$ packets per time slot and the first three rows and the columns of users $s_1$, $s_2$ and $s_3$ in Table \ref{table2}.]{
%  \includegraphics[scale=0.5]{optimal_probabilities_q}
%   \label{fig:subfig3}
%   }\hfill
% \subfigure[Stability region of system $\mathcal{S}$ with $\mathcal{M}_s\!=\!\mathcal{M}_p\!=\!150$. The parameters used to generate the figure are generated using the following functions: $\overline{P}_{{\rm out},j{s_k}}=0.3(1+\frac{j}{\mathcal{M}_p})+0.3\frac{k}{\mathcal{M}_s}$ and $\pi_j=0.3+0.5\frac{j}{\mathcal{M}_p}$.]{
%  \includegraphics[scale=0.5]{high_par}
%   \label{fig:subfig4}
%   }
%\bigskip
% % \subfigure[Stability region of the system $\mathcal{S}$. The parameters used to generate the figure are $\mathcal{M}_s$=4 and parameters in Table \ref{table5}]{
%%  \includegraphics[scale=0.5]{S_four_users_4_bands_6bands.eps}
%%   \label{fig:subfig3}
%%   }
%% \label{fig:subfigureExample}
%{%
% % \subref{},
%%  \subref{} \subref{}
%}
%\end{figure}
%%%%%%%%%%%%%%%%%%%%%%%
%%%%%%%%%%%%%%%%%%%%%%
%%%%CONCLUSION
%%%%%%%%%%%%%%%%%%%%%%%
%%%%%%%%%%%%%%%%%%%%
%%%%%%%%%%%%%%%%%%%%%%%

\end{document}